\documentclass[runningheads]{llncs}
\usepackage[T1]{fontenc}
\usepackage{graphicx}
\usepackage{amssymb}
\usepackage{amsmath}
\usepackage{array}
\usepackage{booktabs}
\usepackage[table,xcdraw]{xcolor}
\usepackage{diagbox}
\usepackage{cite}
\begin{document}
\title{Successor-bispecial strings with minimum Burrows--Wheeler transform runs}
%
 
\author{Vinicius T.~V.~Date\inst{1} \and Leandro M.~Zatesko\inst{2, 1}}
\authorrunning{V.~T.~V.~Date and L.~M.~Zatesko}

\institute{Postgraduate Program in Informatics, Federal University of Paraná \and Federal University of Technology -- Paraná\\
\email{vtvdate@inf.ufpr.br, zatesko@utfpr.edu.br}}
\maketitle

\begin{abstract}
We study \emph{successor-bispecial strings} over an alphabet \(\Sigma\) of size \(\sigma\), a minimal-branching analogue of de~Bruijn strings, and ask how few Burrows--Wheeler transform (BWT) runs are possible. In a de~Bruijn string of order~\(k\), every \((k{-}1)\)-gram has all \(\sigma\) right-extensions; here, every \(k{-}1\)-gram has exactly two right-extensions, determined by a successor rule, which also forces two left-extensions.
For order~\(3\), we construct an explicit family \(B_\sigma^{(3)}\), for every \(\sigma \geq 2\), whose cyclic BWT has \(r_c = \sigma^2 + 2\) runs. A suitable terminated linearization has the same run count, \(r = r_c = \sigma^2 + 2\), while the smallest suffixient set has size \(\chi = 2\sigma^2 + 1\). The ratio \(\chi/r = 2 - 3/(\sigma^2 + 2)\) nearly saturates the known bound \(\chi/r \leq 2\), which we have previously shown to be asymptotically tight. Compared with our earlier general construction, this improves the gap from \(O(1/\sigma)\) to \(O(1/\sigma^2)\).
We also show that the order-\(3\) pattern appears as a blockwise two-row projection of normalized linear-feedback shift register (LFSR) de~Bruijn sequences over \(\mathbb F_q\), when primitive trinomials \(x^3 - x + c\) exist. For higher orders, we prove a general lower bound \(r_c \geq \sigma^{k-1} + 2\) for every \(\sigma \geq 3\) in the exact-length regime and analyze the boundary-merged higher-order candidate using the last-to-first (LF) permutation: it fails for \(k = 4\) and all \(\sigma \geq 3\), while verified \(k = 5\) instances for \(\sigma \in \{3,4\}\) yield \(\chi/r\) ratios exceeding~\(1.96\).
\end{abstract}

\section{Introduction}\label{sec:intro}

The Burrows--Wheeler transform (BWT) orders the rotations or suffixes of a string and records the preceding symbols in the resulting last column; the number of equal-letter runs in this column is denoted by \(r\). We ask how small this run count can be in a restricted de~Bruijn-type regime where every context has two right-extensions determined by a successor rule, which in turn forces two left-extensions. This structural question also has consequences for the measure \(\chi\), the size of a smallest suffixient set~\cite{depuydt2023suffixient}: Navarro et al.~\cite{navarro2025smallest} proved that \(\chi \leq 2r\), while empirical results on genomic data suggest that the gap is often much smaller~\cite{cenzato2024computing}. In~\cite{date2025neartightness} we showed that \(\chi \leq 2r\) is asymptotically tight, both as the alphabet size~\(\sigma\) grows and as the string length grows over a binary alphabet. Here, the large \(\chi/r\) ratio is a consequence of a structural question about BWT runs in a restricted de~Bruijn-type regime.
 
Full de~Bruijn strings of order \(k\)~\cite{lothaire1997} are \(\sigma\)-branching at every context: every \((k{-}1)\)-gram has all \(\sigma\) right-extensions and, since every \(k\)-gram occurs, all \(\sigma\) left-extensions, so every context is bispecial, having multiple
extensions on both sides. At the opposite, minimally branching extreme, we call a cyclic string \emph{successor-bispecial at order~\(k\)} if, for every \((k{-}1)\)-gram~\(u\), the only \(k\)-grams with prefix~\(u\) are \(u \cdot u[0]\) and \(u \cdot ((u[0]{+}1) \bmod \sigma)\). This regime separates the two roles of bispeciality: the successor rule on right-extensions is what makes the suffixient count large, while the induced left-extensions are what make the cBWT blocks orderly. In the length-\(2\sigma^{k-1}\) case, where each allowed \(k\)-gram occurs exactly once, the two rows associated with each \((k{-}1)\)-gram have distinct preceding symbols. Hence, each context block forces at least one internal BWT run boundary, and the natural low-run target is one forced boundary per context, up to global cyclic boundary effects.
We prove (Theorem~\ref{thm:cyclic_lower_bound}) that in the exact-length regime, no successor-bispecial cyclic string over an alphabet of size \(\sigma\geq3\) can have fewer than \(\sigma^{k-1}+2\) cyclic BWT runs. For \(\sigma=2\), exact-length successor-bispecial strings are precisely the binary de~Bruijn strings, and the matching bound \(2^{k-1}+2\), attained by the last column \(1(0011)^{2^{k-2}-1}010\), appears in the proof of~\cite[Theorem~3]{mantaci2017burrows}; our theorem extends it to every alphabet size.

For order~\(3\), we give an explicit family~\(B_\sigma^{(3)}\) for every \(\sigma \geq 2\) with a fully predictable cyclic BWT, attaining this lower bound whenever \(\sigma\geq3\), from which we derive closed-form values for the cyclic run count~\(r_c\), the ordinary run count~\(r\), and the suffixient-set size~\(\chi\). The cyclic BWT has only \(\sigma^2 + 2\) runs, and the ordinary BWT of a suitable terminated linearization has the same number of runs. Since the successor-bispecial structure forces \(\chi = 2\sigma^2 + 1\), we obtain (Theorem~\ref{theorem:ratio})
    \[
        \frac{\chi}{r}
        =
        \frac{2\sigma^2 + 1}{\sigma^2 + 2}
        =
        2 - \frac{3}{\sigma^2 + 2}.
    \]
Compared with the clustered alphabet-growing construction of~\cite{date2025neartightness}, the order-\(3\) family constructed here improves the ratio for every \(\sigma \geq 2\), from \(2\sigma/(\sigma + 1)\) to \((2\sigma^2 + 1)/(\sigma^2 + 2)\), and improves the asymptotic gap from \(O(1/\sigma)\) to \(O(1/\sigma^2)\). This does not settle the fixed-alphabet asymptotic regime; rather, it isolates a restricted de~Bruijn-type setting in which the suffixient count is forced to be large while the BWT remains highly structured.

The order-\(3\) block regularity also has an algebraic interpretation, extending the linear-feedback shift register (LFSR) phenomenon behind the binary construction of~\cite{date2025neartightness}. For prime powers \(q\) such that the trinomial \(x^3 - x + c\) is primitive over \(\mathbb F_q\) for some \(c \in \mathbb F_q\), normalized LFSR de~Bruijn sequences exhibit the same paired cBWT block regularity. More precisely, their normalized cBWT blocks are affine, and a two-row selection from each block recovers the paired low-run pattern of \(B_q^{(3)}\).

We also investigate higher-order analogues through the canonical boundary-merged candidate \(L_{\sigma,k}\), which extends the order-\(3\) cBWT pattern to arbitrary~\(k\). Using the last-to-first (LF) permutation, we prove that this candidate fails for \(k = 4\) and all \(\sigma \geq 3\): the obstruction is to the candidate, not to the existence of successor-bispecial strings of order~\(4\). We also verify that the candidate succeeds sporadically for \(k = 5\) with \(\sigma \in \{3,4\}\), giving \(\chi/r\) ratios exceeding~\(1.96\).

In Section~\ref{sec:prelim} we fix notation and recall standard definitions. In Section~\ref{sec:sbs} we define successor-bispecial strings and prove the exact-length lower bound. In Section~\ref{sec:construction} we present the order-\(3\) construction, and in Section~\ref{sec:measures} we derive the closed-form values of \(r_c\), \(r\), and~\(\chi\). In Section~\ref{sec:lfsr} we develop the LFSR interpretation over finite fields. In Section~\ref{sec:higher_open} we analyze the boundary-merged candidate, prove the \(k{=}4\) obstruction, report the \(k{=}5\) results, and state future directions.

\section{Preliminaries}\label{sec:prelim}

Let \(S\) be a string of length \(|S| = n\) over the alphabet \(\Sigma = \{0,1,\ldots,\sigma{-}1\}\) of size \(|\Sigma|=\sigma\). Strings are \(0\)-indexed. We write \(S[i]\) for the character at position~\(i\) and \(S[i..j]\) for the substring from position~\(i\) to~\(j\) inclusive. A \emph{\(k\)-gram} is a substring of length~\(k\), and a \emph{rotation} of~\(S\) is \(S[i..n{-}1]\cdot S[0..i{-}1]\) for some \(0\leq i< n\).

We use a sentinel symbol \(\$\) satisfying \(\$ \notin \Sigma\) and \(\$ < a\) for all \(a\in\Sigma\). For a string \(w\in\Sigma^*\), we write \(w\$\) for the \emph{terminated string}.

A \emph{cyclic string} of length~\(n\) is a string whose indices are read modulo~\(n\). The \emph{order-\(k\) linearization} of a cyclic string~\(S\) is the string \(\ell_k(S)=S\cdot S[0..k{-}2]\). Its terminated form is \(\ell_k(S)\$\). The appended prefix ensures that all cyclic \((k{-}1)\)-grams and \(k\)-grams appear as ordinary substrings before the sentinel.

The \emph{Burrows--Wheeler transform} \(\mathrm{BWT}(w\$)\)~\cite{burrows1994} is obtained by sorting all rotations of \(w\$\) as rows of a matrix in lexicographic order and reading the last column. A \emph{run} is a maximal contiguous block of equal symbols. We write \(r(w)\) for the number of runs in \(\mathrm{BWT}(w\$)\).

For cyclic strings, the \emph{cyclic Burrows--Wheeler transform} \(\mathrm{cBWT}(S)\) sorts the \(n\) cyclic rotations of~\(S\) lexicographically and reads the last column. We write \(r_c(S)\) for the number of runs in the resulting linear string \(\mathrm{cBWT}(S)\).

Let \(w\in\Sigma^*\). A substring~\(x\) of \(w\$\) is \emph{right-maximal} if it has at least two  \emph{right-extensions} in \(w\$\), i.e.,~two distinct \(a\in\Sigma\cup\{\$\}\) for which \(xa\) occurs as a substring in \(w\$\). The set of right-extensions of \(w\) is
    \[
      E_r(w)
      =
      \{\, xa \mid x \text{ is right-maximal in } w\$,\;
          a\in\Sigma\cup\{\$\},\;
          xa \text{ occurs in } w\$ \,\}.
    \]
The \emph{super-maximal right-extensions} are the elements of \(E_r(w)\) that are not a proper suffix of any other element of \(E_r(w)\). We denote this set by \(S_r(w)\).

A set \(P\subseteq\{0,\ldots,|w|\}\) is \emph{suffixient} ~\cite{depuydt2023suffixient,navarro2025smallest} if for every \(x\in E_r(w)\) there exists \(j\in P\) such that \(x\) is a suffix of \((w\$)[0..j]\). The size of a smallest suffixient set is denoted by \(\chi(w)\). Navarro et al.~\cite{navarro2025smallest} show that \(\chi(w)=|S_r(w)|\).

\begin{example}\label{ex:suffixient}
    Let \(w=\texttt{abaa}\), so \(w\$=\texttt{abaa\$}\). The right-maximal substrings are \(\epsilon\) and \(\texttt{a}\), giving 
    $        E_r(w)=\{\texttt{a},\texttt{b},\texttt{\$},\texttt{aa},\texttt{ab},\texttt{a\$}\}$.
        The super-maximal elements are
    $
        S_r(w)=\{\texttt{aa},\texttt{ab},\texttt{a\$}\}
    $,
    so \(\chi(w)=3\).
\end{example}

For a cyclic string~\(S\) and a string \(u\in\Sigma^{k-1}\), write
    \[
        R_S(u)
        =
        \{\, c\in\Sigma : uc \text{ occurs as a } k\text{-gram of } S \,\}
    \]
for the set of right-extensions of~\(u\) in~\(S\). We call the \((k{-}1)\)-grams of \(S\) \emph{contexts}.

In the terminology of combinatorics on words, a substring with at least two right-extensions is \emph{right-special}; symmetrically, a substring with at least two left-extensions is \emph{left-special}, and a substring that is both is \emph{bispecial}.

\section{Successor-bispecial strings and the cyclic lower bound}\label{sec:sbs}

Our definition below fixes only the right-extensions of each context, but Proposition~\ref{prop:left_extensions} shows that this already forces two fixed left-extensions. We therefore call the resulting strings successor-bispecial. We reserve the term right-maximal for substrings of the terminated string \(w\$\), following the suffixient-set definition above.

\begin{definition}\label{def:srs}
For \(\sigma\geq2\), a cyclic string~\(S\) over \(\Sigma=\{0,\ldots,\sigma{-}1\}\) is \emph{successor-bispecial at order~\(k\)} if, for every \(u\in\Sigma^{k-1}\), \(R_S(u) = \{\, u[0],\; (u[0]{+}1) \bmod \sigma \,\}\).
\end{definition}

From this definition, in a successor-bispecial cyclic string~\(S\) at order~\(k\), every \((k{-}1)\)-gram over~\(\Sigma\) occurs, since \(R_S(u)\) is required to be nonempty for every \(u\). The set of occurring \(k\)-grams is exactly
    \[
        \{\, uc \mid u\in\Sigma^{k-1},\;
            c\in\{u[0], (u[0]{+}1) \bmod\sigma\}\,\}.
    \]
This set has size \(2\sigma^{k-1}\). In particular, if a successor-bispecial cyclic string at order~\(k\) has length \(2\sigma^{k-1}\), then each allowed \(k\)-gram occurs exactly once, since the string has exactly \(2\sigma^{k-1}\) \(k\)-gram positions and no disallowed \(k\)-gram can occur. We call this the \emph{exact-length regime}.

The term \emph{successor} reflects that the two extensions of each context~\(u\) are the symbol \(u[0]\) and its cyclic successor \((u[0]{+}1) \bmod \sigma\).

\begin{proposition}\label{prop:left_extensions}
Let \(S\) be successor-bispecial at order \(k\). Then every context \(u\) of \(S\) is left-special. More precisely, if \(t = u[k-2]\), then the two left-extensions of \(u\) are \(t - 1\) and \(t\), modulo \(\sigma\).
\end{proposition}
\begin{proof}
    Let \(u\) be a context of \(S\), and let \(t = u[k-2]\). If \(c\) precedes an occurrence of \(u\), then \(cu[0..k-3]\) is a context starting with \(c\), and Definition~\ref{def:srs} gives \(t \in \{c, (c + 1) \bmod \sigma\}\). Thus \(c \in \{t, (t - 1) \bmod \sigma\}\).

    Conversely, both possibilities occur. Indeed, for \(c = t\) and for \(c = (t - 1) \bmod \sigma\), the context \(cu[0..k-3]\) has \(t\) as one of its prescribed right-extensions. Hence both \(tu\) and \(((t - 1) \bmod \sigma)u\) occur. Therefore the left-extensions of \(u\) are exactly \(t - 1\) and \(t\), modulo \(\sigma\). \qed
\end{proof}

Thus, the right-extension successor rule also forces a successor rule for left-extensions. We next record the lower bound that makes the run count targeted in the rest of the paper optimal.

We use the standard cBWT/LF characterization: a string is the cBWT of a single primitive cyclic string, i.e., one that is not a nontrivial repetition, if and only if its induced last-to-first (LF) permutation (mapping each position of the last column to the position of the same symbol occurrence in the first column) is a single cycle~\cite{likhomanov2011two}.

\begin{theorem}\label{thm:cyclic_lower_bound}
    Let \(\sigma \geq 3\), \(k \geq 3\), and let \(S\) be successor-bispecial at order \(k\), with \(|S| = 2\sigma^{k-1}\). Then \(r_c(S) \geq \sigma^{k-1} + 2\).
\end{theorem}
\begin{proof}
    In the sorted rotation matrix, rows whose rotations start with the same context \(u \in \Sigma^{k-1}\) are consecutive. Since \(|S| = 2\sigma^{k-1}\), each context occurs exactly twice, so the cBWT is partitioned into \(\sigma^{k-1}\) blocks of length two. By Proposition~\ref{prop:left_extensions}, the block of a context \(u\) contains the two distinct symbols \((t - 1) \bmod \sigma\) and \(t\), where \(t = u[k-2]\). Thus the blocks contribute \(\sigma^{k-1}\) internal run boundaries. Let \(h\) be the number of additional run boundaries occurring between consecutive context blocks in the linear cBWT. Then \(r_c(S)=1+\sigma^{k-1}+h\).

    We show that \(h\neq0\). Suppose, for contradiction, that \(h=0\). Then, no run boundary occurs between consecutive context blocks. Consecutive contexts in lexicographic order have last symbols \(t\) and \((t+1)\bmod\sigma\). Since \(\sigma\geq3\), the corresponding block contents share exactly one symbol. Therefore, each adjacent pair of blocks must meet at this shared symbol. It follows that each non-final block with last symbol \(t\) ends with \(t\), and each non-first block with last symbol \(t\) begins with \((t-1)\bmod\sigma\). In particular, the final block, whose context is \((\sigma-1)^{k-1}\), begins with \(\sigma-2\) and therefore ends with \(\sigma-1\), so \(\mathrm{cBWT}(S)\) ends with the largest alphabet symbol.

   The string \(S\) is primitive, since otherwise each \(k\)-gram would appear more than once, contradicting the exact-length condition. Hence, by the cBWT/LF characterization recalled above, the induced LF permutation must be a single cycle. However, since the cBWT ends with the largest symbol, its last occurrence maps to the last position of the first column, giving a fixed point. This contradiction shows that \(h \geq 1\), and therefore \(r_c(S) \geq 1 + \sigma^{k-1} + 1 = \sigma^{k-1} + 2\). \qed
\end{proof}

For \(\sigma = 2\), the successor-bispecial condition coincides with the binary de~Bruijn condition, and the corresponding bound follows from the binary de~Bruijn BWT analysis of Higgins~\cite{higgins2012burrows} and Mantaci et al.~\cite{mantaci2017burrows}. 

Together, these bounds identify \(\sigma^{k-1}+2\) as the optimal cyclic run-count target in the exact-length successor-bispecial regime. The construction in the next section attains this value for \(k=3\).

\section{The construction} \label{sec:construction}
We now give the order-\(3\) construction that attains the lower bound of
Theorem~\ref{thm:cyclic_lower_bound}.

We define a family of cyclic strings \(B_\sigma^{(3)}\), indexed by alphabet size \(\sigma \geq 2\). The design principle is that block \(a\) lists all pairs \((a,b)\) to cover every context with first symbol~\(a\). The repeated terminal pair at the end of each block creates an odd-position crossing into the next block, which supplies the second extension \((a+1) \bmod \sigma\) for every context~\(ab\). The string \(B_\sigma^{(3)}\) is over the alphabet \(\Sigma = \{0,1,\ldots,\sigma-1\}\) and is defined as the concatenation of \(\sigma\) \emph{blocks}, one for each symbol in \(0,1,\dotsc,\sigma-1\), in this order. For \(a\in\{0,1,\dotsc,\sigma-2\}\), the \(a\)-th block is constructed by concatenating the pairs \((a,0),(a,1),\dotsc,(a,\sigma-1),(a,\sigma-1)\), where we write \((a,b)\) to denote the \(2\)-gram \(ab\), i.e.~the two-character string with first character \(a\) and second character \(b\). Note that the pair \((a,\sigma-1)\) is repeated at the end of each block for \(a\leq\sigma-2\). For \(a=\sigma-1\), the block consists only of the pair \((\sigma-1,\sigma-1)\).

\begin{example}
    The string \(B^{(3)}_4\): \texttt{00010203031011121313202122232333}.
\end{example}

\begin{remark}
    The string \(B_\sigma^{(3)}\) has length \(2\sigma^2\): each block \(a < \sigma-1\) contributes \(2(\sigma+1)\) characters, and block \(\sigma-1\) contributes \(2\) characters, yielding \((\sigma-1) \cdot 2(\sigma+1) + 2 = 2(\sigma^2-1) + 2 = 2\sigma^2\).
\end{remark}

Recall that \(R_S(u)\) denotes the set of right-extensions of a context \(u\) in the cyclic string \(S\). To prove that \(B_\sigma^{(3)}\) is successor-bispecial at order~\(3\), we must show that every length-\(2\) context has the two prescribed extensions. Thus, for every first symbol \(a\in\Sigma\) and every second symbol \(b\in\Sigma\), we prove
\[
    R_{B_\sigma^{(3)}}(ab)=\{a,(a+1)\bmod\sigma\}.
\]
The role of \(b\) is to range over all contexts whose first symbol is \(a\); the successor rule then says that the two extensions depend only on this first symbol.

\begin{lemma}\label{lemma:extensions}
    For every \(ab \in \Sigma^2\), the \(2\)-gram \(ab\) occurs exactly twice in~\(B^{(3)}_\sigma\), and \(R_{B^{(3)}_\sigma}(ab) = \{a,\, (a{+}1) \bmod \sigma\}\). Consequently, the occurring cyclic \(3\)-grams are exactly the \(abc\) with \(c \in \{a,\, (a{+}1) \bmod \sigma\}\), and there are \(2\sigma^2\) of them.
\end{lemma}

\begin{proof}
    Since each block begins at an even position and consists of pairs, every listed pair starts at an even index; we call an occurrence \emph{even} or \emph{odd} according to the parity of its starting position. It is enough to identify two occurrences of each \(2\)-gram: these \(2\sigma^2\) occurrences account for all cyclic positions, so no further occurrences are possible.
    
    \smallskip\noindent
    \textit{Case \(a < \sigma{-}1\), \(b < \sigma{-}1\).}
    The even occurrence is the pair \((a,b)\) in block~\(a\), followed by \((a,b{+}1)\), and has extension~\(a\). The odd occurrence is the pair \((b,a)\) in block~\(b\), followed by \((b,a{+}1)\), and has extension~\(a{+}1\).
    
    \smallskip\noindent
    \textit{Case \(a < \sigma{-}1\), \(b = \sigma{-}1\).}
    The repeated terminal pair \((a,\sigma{-}1)\) in block~\(a\) gives two even occurrences: the first copy is followed by the second and has extension~\(a\), while the second copy is followed by block~\(a{+}1\) and has extension~\(a{+}1\). Block~\(\sigma{-}1\) contains no pair \((\sigma{-}1,a)\), so no odd occurrence exists.
    
    \smallskip\noindent
    \textit{Case \(a = \sigma{-}1\), \(b < \sigma{-}1\).}
    Block~\(\sigma{-}1\) contains only \((\sigma{-}1,\sigma{-}1)\), so no even occurrence exists. The repeated pair \((b,\sigma{-}1)\) in block~\(b\) gives an odd occurrence of \((\sigma{-}1)b\) with extension~\(\sigma{-}1\), while the block boundary \(b{-}1 \to b\), with indices modulo~\(\sigma\), gives an odd occurrence with extension~\(0\).
    
    \smallskip\noindent
    \textit{Case \(a = b = \sigma{-}1\).}
    The even occurrence is the pair \((\sigma{-}1,\sigma{-}1)\) in block~\(\sigma{-}1\), followed cyclically by block~\(0\), and has extension~\(0\). The odd occurrence is at the end of block~\(\sigma{-}2\) and has extension~\(\sigma{-}1\).
    
    \smallskip\noindent
    In every case, \(R_{B^{(3)}_\sigma}(ab) = \{a,\,(a{+}1) \bmod \sigma\}\). \qed
\end{proof}

Thus \(B_\sigma^{(3)}\) is successor-bispecial at order~\(3\), with exact length \(2\sigma^2\).

\begin{remark}
For \(\sigma = 2\), the allowed set is the full set of \(2^3\) binary \(3\)-grams, so \(B_2^{(3)}\) is a de~Bruijn sequence of order~\(3\).
\end{remark}

\section{Measuring the string family} \label{sec:measures}

In this section, we measure the strings \(B_\sigma^{(3)}\). We start by counting the number of runs in the cBWT.

\begin{lemma}\label{lemma:c_runs}
    The cBWT for \(B_\sigma^{(3)}\) has \(\sigma^2 + 2\) runs.
\end{lemma}
\begin{proof}
    Since all \(2\sigma^2\) cyclic \(3\)-grams of \(B_\sigma^{(3)}\) are distinct, the relative order of the cyclic rotations is determined by their first three symbols. We divide the rows into groups with a common prefix \(ab\). By Lemma~\ref{lemma:extensions}, such groups contain the \(3\)-grams \(aba\) and \(ab((a+1)\bmod\sigma)\).

    For a context \(ab\), Proposition~\ref{prop:left_extensions} shows that the corresponding cBWT block contains the two symbols \((b - 1) \bmod \sigma\) and \(b\). It remains to determine their order.
    
    First suppose \(a < \sigma - 1\). Then the two rows in the block are ordered as \(aba\) and \(ab(a + 1)\). The occurrence of \(aba\) starts at the listed pair \((a,b)\) in block~\(a\), so its preceding symbol is \((b - 1) \bmod \sigma\). The occurrence of \(ab(a + 1)\) starts between the pair \((b,a)\) and the following pair \((b,a+1)\) in block~\(b\), so its preceding symbol is \(b\); when \(b = \sigma - 1\), the same two preceding symbols come from the repeated terminal pair \((a,\sigma - 1)\) in block~\(a\). Hence, as \(b\) ranges from \(0\) to \(\sigma - 1\), the contribution of the groups with first context symbol \(a\) is
    \[
        (\sigma - 1)0011\cdots(\sigma - 2)(\sigma - 2)(\sigma - 1).
    \]

    For first context symbol \(\sigma - 1\), all groups except the last contribute as above. The final group \((\sigma - 1)(\sigma - 1)\) has sorted \(3\)-grams
    \[
        (\sigma-1)(\sigma-1)0,\quad
        (\sigma-1)(\sigma-1)(\sigma-1),
    \]
    with preceding characters \(\sigma-1\) and \(\sigma-2\), respectively. Therefore, the contribution of the groups with the first context symbol \(\sigma-1\) is
    \[
        (\sigma-1)0011\cdots
        (\sigma-3)(\sigma-3)(\sigma-2)(\sigma-1)(\sigma-2).
    \]

    Concatenating the contributions of all first context symbols, we obtain
    \[
    \begin{aligned}
    \mathrm{cBWT}(B_\sigma^{(3)})
    ={}&
    (\sigma - 1)
    (0011\cdots(\sigma - 1)(\sigma - 1))^{\sigma - 1} \\
    &\quad
    0011\cdots
    (\sigma - 3)(\sigma - 3)(\sigma - 2)(\sigma - 1)(\sigma - 2).
    \end{aligned}
    \]
    Here and below, empty ranges are omitted. This gives one leading run, \(\sigma\) runs for each first context symbol \(a < \sigma - 1\), and \(\sigma + 1\) runs for the final first context symbol. Hence \(r_c(B_\sigma^{(3)}) = 1 + \sigma(\sigma - 1) + (\sigma + 1) = \sigma^2 + 2\). \qed
\end{proof}

For \(\sigma\geq3\), this matches the lower bound of Theorem~\ref{thm:cyclic_lower_bound} with \(k=3\), so the cyclic run count of \(B_\sigma^{(3)}\) is optimal.

\begin{example}\label{ex:cbwt}
    For \(\sigma = 3\),
    \(B_3^{(3)} = \texttt{000102021011121222}\) has
    \[
        \mathrm{cBWT}(B_3^{(3)})
        =
        \texttt{200112200112200121},
    \]
    which has \(\sigma^2 + 2 = 11\) runs.
\end{example}

\begin{lemma}\label{lemma:bwt_runs}
    Let \(t = \ell_3(R)\), where \(R\) is the rotation of \(B_\sigma^{(3)}\) starting with \((\sigma - 1)^3\). Then \(\mathrm{BWT}(t\$)\) has \(\sigma^2 + 2\) runs.
\end{lemma}

\begin{proof}
    Since \(R\) starts with \((\sigma - 1)^3\), we have \(t=R(\sigma - 1)(\sigma - 1)\).

    The rows of the rotation matrix of \(t\$\) whose starts lie inside \(R\) have the same relative order as the corresponding cyclic rotations, since their first three symbols are distinct. Thus they reproduce the cBWT pattern from Lemma~\ref{lemma:c_runs}, except that the row \(t\$\), whose prefix is \((\sigma - 1)^3\), has BWT character \(\$\) instead of the final cyclic character \(\sigma - 2\).

    The additional rows begin with \(\$\), \((\sigma - 1)\$\), and \((\sigma - 1)(\sigma - 1)\$\). They are inserted before the rows with prefixes \(00\), \((\sigma - 1)0\), and \((\sigma - 1)(\sigma - 1)0\), respectively; their BWT characters are \(\sigma - 1\), \(\sigma - 1\), and \(\sigma - 2\).
    \[
        \begin{aligned}
        \mathrm{BWT}(t\$)
        ={}&
        \boldsymbol{(\sigma - 1)}(\sigma - 1)
        (0011\cdots(\sigma - 1)(\sigma - 1))^{\sigma - 1}\\
        &\quad
        \boldsymbol{(\sigma - 1)}0011\cdots
        (\sigma - 3)(\sigma - 3)
        (\sigma - 2)\boldsymbol{(\sigma - 2)}
        (\sigma - 1)\boldsymbol{\$}.
        \end{aligned}
    \]
    The three inserted characters and the sentinel replacing the final cyclic \(\sigma-2\) are shown in bold. This string has one initial \((\sigma - 1)\) run, \(\sigma\) runs for each of the \(\sigma - 1\) regular blocks, and \(\sigma + 1\) runs in the final block. Hence
    \(r(t)=1+\sigma(\sigma - 1)+(\sigma + 1)=\sigma^2 + 2.\)\qed
\end{proof}

\begin{lemma}\label{lemma:smallest_set}
Let \(t=\ell_3(R)\) for any rotation \(R\) of \(B_\sigma^{(3)}\). Then \(\chi(t)=2\sigma^2+1\).
\end{lemma}
\begin{proof}
    By \cite{navarro2025smallest}, \(\chi(t)=|S_r(t)|\), so it suffices to count the super-maximal right-extensions of \(t\$\).
    By Lemma~\ref{lemma:extensions}, the cyclic \(3\)-grams of \(B_\sigma^{(3)}\), and hence the \(3\)-grams occurring in \(t\), are exactly \(\{\,ab c : ab\in\Sigma^2,\ c\in\{a,(a+1)\bmod\sigma\}\,\}\).
    
    There are \(2\sigma^2\) such \(3\)-grams.
    Moreover, Lemma~\ref{lemma:extensions} shows that each \(2\)-gram occurs exactly twice and realizes the two distinct extensions, so each \(3\)-gram occurs exactly once. Hence, any string of length at least \(3\) occurs at most once, because its length-\(3\) prefix occurs only once. Therefore, no string of length at least \(3\) is right-maximal, and no non-sentinel super-maximal right-extension has length at least \(4\).

    Conversely, every occurring \(3\)-gram has a length-\(2\) prefix with the two extensions given by Lemma~\ref{lemma:extensions}, so it belongs to \(E_r(t)\). Since no longer non-sentinel element of \(E_r(t)\) exists, these \(2\sigma^2\) \(3\)-grams are precisely the super-maximal right-extensions of \(t\$\) not involving the sentinel. It remains to account for right-extensions ending in \(\$\). Let \(p\) be the length-\(2\) suffix of \(t\) immediately preceding the sentinel. The same \(2\)-gram \(p\) occurs at the beginning of \(t\), followed by a symbol of \(\Sigma\), and at the end of \(t\), followed by \(\$\). Hence \(p\) is right-maximal in \(t\$\), and \(p\$\in E_r(t)\). If \(x\$\in E_r(t)\) with \(|x|\geq3\), then \(x\) occurs only as the terminal suffix of \(t\), because every \(3\)-gram occurs once; hence \(x\) is not right-maximal. Any shorter sentinel-ending right-extension is a suffix of \(p\$\).

    Hence \(|S_r(t)|=2\sigma^2+1\), and \(\chi(t)=2\sigma^2+1\). \qed
\end{proof}

\begin{theorem}\label{theorem:ratio}
    Let \(R\) be the rotation of \(B_\sigma^{(3)}\) that starts with \((\sigma -1)^3\), and let \(t=\ell_3(R)\). Then 
    \[
        \frac{\chi(t)}{r(t)}
        =
        \frac{2\sigma^2+1}{\sigma^2+2}.
    \]
    In particular, this ratio approaches \(2\) as \(\sigma\to\infty\).
\end{theorem}
\begin{proof}
    This follows immediately from Lemmas~\ref{lemma:bwt_runs} and~\ref{lemma:smallest_set}. \qed
\end{proof}

For every \(\sigma \geq 2\), our construction improves on the clustered
alphabet-growing construction of~\cite{date2025neartightness}, whose ratio
is \(2\sigma/(\sigma + 1)\).

\section{Algebraic interpretation of the paired cBWT blocks}\label{sec:lfsr}

We now give an algebraic interpretation of the optimal paired cBWT block pattern of Section~\ref{sec:measures}. BWT-based views of de~Bruijn words have also been studied~\cite{higgins2012burrows,fici2025generalized}. In our previous work~\cite{date2025neartightness}, the binary construction was tied directly to the cBWT structure of a linear-feedback shift register (LFSR) de~Bruijn sequence. Here we show that an analogous phenomenon persists over finite fields: under a hypothesis on primitive trinomials, normalized LFSR de~Bruijn sequences have affine cBWT blocks, and a two-row selection extracts the order-\(3\) paired pattern.

Let \(q\) be a prime power; throughout this section, the alphabet is \(\mathbb F_q\), so \(\sigma = q\). Assume that there exists \(c \in \mathbb F_q\) such that \(x^3 - x + c\) is primitive over \(\mathbb F_q\), and fix one such \(c\)~\cite{lidl1997finite}. Fix an alphabet order \(\beta_0 < \cdots < \beta_{q - 1}\) on \(\mathbb F_q\), write \(\lambda = \beta_{q - 1}\), and let \(\operatorname{pred}(\beta_i) = \beta_{(i - 1) \bmod q}\) be the cyclic predecessor in this order. All arithmetic below is in \(\mathbb F_q\).
The associated recurrence \(s_{n+3} = s_{n+1} - cs_n\) has period \(q^3 - 1\) on the nonzero states, where the state at time \(n\) is the triple \((s_n, s_{n+1}, s_{n+2})\).

We also set \(\eta = \lambda - \operatorname{pred}(\lambda)\). Since \(\eta \neq 0\), the state \((0, 0, \eta)\) is nonzero and therefore appears on the cycle, so the factor \(00\eta\) occurs. To obtain a de~Bruijn sequence, we insert the missing zero state at this occurrence, replacing the local factor \(00\eta\) by \(000\eta\). After this insertion, define \(d_i=\lambda-s_{-i}\), and let \(D_q\) be the resulting cyclic sequence. This reversal and complementation normalizes the sequence so that, away from the insertion seam, the predecessor of a row with prefix \((x,y,z)\) is independent of \(x\).

\begin{lemma}\label{lem:lfsr_predecessor}
Let \(i\) be a position of \(D_q\) with cyclic indices, such that none of the positions \(i-1, \ldots, i+2 \) lies in the inserted occurrence of the zero state. Write \(D_q[i..i + 2] = (x,y,z)\). Then \(D_q[i - 1] = P(x,y,z) = y + c(\lambda - z)\). In particular, this predecessor is independent of \(x\).
\end{lemma}

\begin{proof}
Since \(d_j = \lambda - s_{-j}\), we have \(s_{-i-1} = \lambda - y\) and \(s_{-i-2} = \lambda - z\). Away from the insertion seam, the recurrence applies at index \(-i-2\), giving \(s_{-i+1} = s_{-i-1} - cs_{-i-2}\).

Therefore, the preceding symbol is
\[
    d_{i-1}
    =
    \lambda - s_{-i+1}
    =
    \lambda - \bigl((\lambda - y) - c(\lambda - z)\bigr)
    =
    y + c(\lambda - z). \tag*{\qed}
\]
\end{proof}

Thus, away from the insertion seam, inside every cBWT block whose rows have a common prefix \((a,b)\), the row indexed by the third symbol \(d\) contributes \(b+c(\lambda-d)\) to the cBWT. Hence, each unaffected block is affine in \(d\) and independent of \(a\).

\begin{proposition}\label{prop:lfsr_mask}
For each \(b\in\mathbb F_q\), define \(
    M_b=\left\{\lambda,\;\lambda-c^{-1}\bigl(\operatorname{pred}(b)-b\bigr)\right\},
\)
where \(\operatorname{pred}(b)\) is taken in the fixed alphabet order, while subtraction and inversion are field operations. Then, for every prefix block \((a,b)\) away from the zero-state insertion seam, selecting the rows with \(d\in M_b\) gives exactly the two cBWT symbols \(\operatorname{pred}(b),b\), in this order.
\end{proposition}

\begin{proof}
Fix a prefix block \((a,b)\), and let \(d\) be the third coordinate of a row, so the row prefix is \((a,b,d)\). The cBWT symbol of this row is its preceding symbol; by Lemma~\ref{lem:lfsr_predecessor}, this is
\(P(a,b,d)=b+c(\lambda-d)\).

We select the rows producing the two paired symbols \(\operatorname{pred}(b)\) and \(b\), which are the symbols of the order-\(3\) block pattern from Lemma~\ref{lemma:c_runs}. The equation \(b+c(\lambda-d)=b\) gives \(d=\lambda\), while \(b+c(\lambda-d)=\operatorname{pred}(b)\) gives \(d=\lambda-c^{-1}(\operatorname{pred}(b)-b)\). These are exactly the two
elements of \(M_b\).

Since \(\operatorname{pred}(b)\neq b\), the second selected value is distinct from \(\lambda\). Since the first three symbols distinguish the rows, within a fixed prefix block \((a,b)\), rows are ordered by the third coordinate \(d\), and \(\lambda\) is the largest symbol in the chosen alphabet order. Hence, the row producing \(\operatorname{pred}(b)\) appears before the row producing \(b\). Thus, the selected rows contribute \(\operatorname{pred}(b),b\), in this order. \qed
\end{proof}

As the prefix blocks are read in lexicographic order, the selected rows give the paired pattern
\[
    \beta_{q-1}
    (\beta_0\beta_0\beta_1\beta_1\cdots\beta_{q-1}\beta_{q-1})^{q-1}
    \beta_0\beta_0\beta_1\beta_1\cdots
    \beta_{q-3}\beta_{q-3}\beta_{q-2}\beta_{q-1}\beta_{q-2}.
\]
This is the paired block structure underlying the cBWT of \(B_q^{(3)}\). When \(q\) is prime and the alphabet is ordered as \(0<1<\cdots<q-1\), this reduces to the shorthand \(\lambda=q-1\) and \(\operatorname{pred}(b)=(b-1)\bmod q\).

The predecessor formula and Proposition~\ref{prop:lfsr_mask} hold away from the seam introduced by the zero-state insertion. At the seam, the insertion replaces the local factor \(00\eta\) by \(000\eta\), so only the two cBWT rows adjacent to the inserted zero state can differ from the affine formula. Under the normalization \(d_i=\lambda-s_{-i}\), the inserted zero state becomes \(\lambda\lambda\lambda\), so the only affected prefix block is \((\lambda,\lambda)\). Moreover, since \(\eta=\lambda-\operatorname{pred}(\lambda)\), the two seam rows have third coordinates
\[
    d=\lambda
    \quad\text{and}\quad
    d=\lambda-c^{-1}(\operatorname{pred}(\lambda)-\lambda),
\]
which are exactly the two values selected by \(M_\lambda\). Thus every selected block contributes \(\operatorname{pred}(b),b\), except the final block \((\lambda,\lambda)\), which contributes \(\lambda,\operatorname{pred}(\lambda)\).

Consequently, after the two-row selection and the single seam correction, the selected LFSR cBWT rows form exactly the boundary-merged paired pattern of \(\mathrm{cBWT}(B_q^{(3)})\) from Lemma~\ref{lemma:c_runs}. Equivalently, \(L_{q,3}\), the boundary-merged candidate of Section~\ref{sec:higher_open} evaluated at \(k=3\), coincides with this cBWT pattern; the optimal order-\(3\) last column therefore has both an explicit combinatorial realization and an LFSR origin.

\section{Higher-order optimal candidates and LF obstructions}\label{sec:higher_open}

Theorem~\ref{thm:cyclic_lower_bound} identifies \(\sigma^{k-1}+2\) as the minimum possible cyclic run count in the exact-length successor-bispecial regime, for \(\sigma \geq 3\). We now ask when this lower bound can be attained. The natural candidate is the boundary-merged last column \(L_{\sigma,k}\), obtained by arranging the two-row context blocks so that only one inter-block boundary creates an additional run boundary. It has exactly the lower-bound number of runs; the remaining question is whether the induced last-to-first (LF) permutation is a single cycle; we call this condition \emph{LF-validity}.

For \(\sigma \geq 2\) and \(k \geq 3\), define the \emph{canonical boundary-merged candidate}
\[
\begin{aligned}
    L_{\sigma,k}
    ={}&
    (\sigma{-}1)
    (0\,0\,1\,1\,\cdots\,(\sigma{-}1)(\sigma{-}1))^{\sigma^{k-2}-1} \\
    &\quad
    0\,0\,1\,1\,\cdots\,
    (\sigma{-}3)(\sigma{-}3)\,
    (\sigma{-}2)(\sigma{-}1)(\sigma{-}2).
\end{aligned}
\]
It has length \(2\sigma^{k-1}\), and empty ranges are omitted. For \(\sigma=2\), it reduces to the extremal pattern of~\cite[Theorem~3]{mantaci2017burrows}. As a linear string, \(L_{\sigma,k}\) has \(\sigma^{k-1}+2\) runs. By Theorem~\ref{thm:cyclic_lower_bound}, this run count is best possible for \(\sigma \geq 3\). Thus, whenever \(L_{\sigma,k}\) is the cBWT of a successor-bispecial cyclic string, it realizes the minimum possible cyclic run count.

We now apply the LF-validity test to this candidate. Recall that LF maps a position \(i\) in a last column \(L\) to the position of the same symbol occurrence in the first column. If \(c=L[i]\), then
\[
    \mathrm{LF}(i) = C[c] + \mathrm{rank}(c,i),
\]
where \(C[c]\) is the number of symbols in \(L\) smaller than \(c\), and \(\mathrm{rank}(c,i)\) counts occurrences of \(c\) strictly before position \(i\).

For \(L_{\sigma,4}\), each symbol occurs \(2\sigma^2\) times. Since the alphabet is ordered \(0 < 1 < \cdots < \sigma - 1\), we have \(C[c] = c \cdot 2\sigma^2\). Therefore, with zero-based positions,
\[
    \mathrm{LF}(i)
    =
    L_{\sigma,4}[i]\cdot 2\sigma^2
    +
    \mathrm{rank}(L_{\sigma,4}[i],i).
\]

\begin{proposition}\label{prop:k4}
    For every \(\sigma \geq 3\), the candidate \(L_{\sigma,4}\) is not a valid cBWT.
\end{proposition}
\begin{proof}
    It suffices to exhibit a proper LF cycle. For \(\sigma = 3\), the LF orbit \(4 \to 19 \to 6 \to 38 \to 13 \to 4\) has length \(5 < 54 = |L_{3,4}|\), so the LF permutation is not a single cycle.

    For \(\sigma \geq 4\), Table~\ref{tab:cycle} gives a symbolic LF-cycle certificate. In each row, \(L[i] = L_{\sigma,4}[i]\), and the displayed \(L[i]\) and \(\rho\) values are obtained directly from the block form of \(L_{\sigma,4}\); substituting them into \(\mathrm{LF}(i) = 2L[i]\sigma^2 + \rho\) gives the corresponding next step of the orbit. All displayed positions lie in \([0,2\sigma^3)\) for \(\sigma \geq 4\).
    
    The labels H, R, and T denote the head, repeating, and tail parts of the displayed orbit. The rows R0--R5 depend on a parameter \(d\in\{1,\ldots,\sigma-3\}\); in the orbit, \(d\) decreases from \(\sigma-3\) to \(1\).

    \begin{table}[t]
    \centering
    \caption{LF-cycle certificate for \(L_{\sigma,4}\), \(\sigma \geq 4\). Here \(\rho = \mathrm{rank}(L_{\sigma,4}[i],i)\).}
    \label{tab:cycle}
    \scriptsize
    \setlength{\tabcolsep}{2pt}
    \renewcommand{\arraystretch}{0.86}
    \resizebox{\textwidth}{!}{%
    \begin{tabular}{@{}clcl@{\qquad}clcl@{}}
    \toprule
    Step & Position \(i\) & \(L[i]\) & \(\rho\)
    &
    Step & Position \(i\) & \(L[i]\) & \(\rho\) \\
    \midrule
    H0 & \(4\) & \(1\) & \(1\)
    &
    H1 & \(2\sigma^2{+}1\) & \(0\) & \(2\sigma\) \\
    
    H2 & \(2\sigma\) & \(\sigma{-}1\) & \(2\)
    &
    H3 & \(2\sigma^3{-}2\sigma^2{+}2\) & \(0\) & \(2\sigma^2{-}2\sigma{+}1\) \\
    
    H4 & \(2\sigma^2{-}2\sigma{+}1\) & \(0\) & \(2\sigma{-}2\)
    &
    H5 & \(2\sigma{-}2\) & \(\sigma{-}2\) & \(1\) \\
    
    H6 & \(2\sigma^3{-}4\sigma^2{+}1\) & \(0\) & \(2\sigma^2{-}4\sigma\)
    &
    H7 & \(2\sigma^2{-}4\sigma\) & \(\sigma{-}1\) & \(2\sigma{-}4\) \\
    
    H8 & \(2\sigma^3{-}2\sigma^2{+}2\sigma{-}4\) & \(\sigma{-}3\) & \(2\sigma^2{-}2\sigma{+}1\)
    &
    H9 & \(2\sigma^3{-}4\sigma^2{-}2\sigma{+}1\) & \(0\) & \(2\sigma^2{-}4\sigma{-}2\) \\
    
    H10 & \(2\sigma^2{-}4\sigma{-}2\) & \(\sigma{-}2\) & \(2\sigma{-}5\)
    &
    H11 & \(2\sigma^3{-}4\sigma^2{+}2\sigma{-}5\) & \(\sigma{-}3\) & \(2\sigma^2{-}4\sigma\) \\
    
    R\(0^*\) & \(2\sigma^3{-}4\sigma^2{-}4\sigma\) & \(\sigma{-}1\) & \(2\sigma^2{-}4\sigma{-}4\)
    &
    R1 & \(2\sigma^2d{+}2\sigma d{+}4\sigma^2{+}2d{+}2\) & \(d\) & \(2\sigma d{+}4\sigma{+}2d{+}1\) \\
    
    R2 & \(2\sigma^2d{+}2\sigma d{+}4\sigma{+}2d{+}1\) & \(d\) & \(2\sigma d{+}2d{+}4\)
    &
    R3 & \(2\sigma^2d{+}2\sigma d{+}2d{+}4\) & \(d{+}1\) & \(2\sigma d{+}2d{+}1\) \\
    
    R4 & \(2\sigma^2d{+}2\sigma d{+}2\sigma^2{+}2d{+}1\) & \(d\) & \(2\sigma d{+}2\sigma{+}2d\)
    &
    R5 & \(2\sigma^2d{+}2\sigma d{+}2\sigma{+}2d\) & \(d{-}1\) & \(2\sigma d{+}2d{+}3\) \\
    
    R0 & \(2\sigma^2d{+}2\sigma d{+}2\sigma{+}2d{+}5\) & \(d{+}2\) & \(2\sigma d{+}2d{+}2\)
    &
    T0 & \(2\sigma{+}5\) & \(2\) & \(2\) \\
    
    T1 & \(4\sigma^2{+}2\) & \(0\) & \(4\sigma{+}1\)
    &
    T2 & \(4\sigma{+}1\) & \(0\) & \(4\) \\
    \bottomrule
    \end{tabular}%
    }
    \end{table}
    
    The cycle consists of the twelve head steps H0--H11, the bridge step R\(0^*\), the repeating phase with parameter \(d = \sigma - 3,\sigma - 4,\ldots,1\), and the tail T0--T2, which closes the orbit back to H0. The bridge step R\(0^*\) enters R1 with \(d = \sigma - 3\). For each fixed \(d\), the orbit follows R1--R5, and R5 with parameter \(d\) maps to R0 with parameter \(d - 1\). After R5 with \(d = 1\), the orbit enters the tail. Therefore, the total length is \(12 + 1 + 5 + 6(\sigma - 4) + 3 = 6\sigma - 3\).
    
    The listed positions are pairwise distinct for \(\sigma \geq 4\): within the repeating phase, this follows from the common term \(2(\sigma^2 + \sigma + 1)d\) and distinct offsets, and the head, bridge, and tail positions are checked directly from the displayed formulas. Since \(6\sigma - 3 < 2\sigma^3 = |L_{\sigma,4}|\), this is a proper LF cycle. Hence, the LF permutation is not a single cycle, and \(L_{\sigma,4}\) is not a valid cBWT. \qed
\end{proof}

For \(\sigma = 2\), the candidate \(L_{2,4}\) has a single LF cycle, consistent with~\cite{date2025neartightness}. For \(k = 5\), we computationally evaluated the \(\mathrm{LF}\) permutation of \(L_{\sigma,5}\), checked that it is a single cycle for \(\sigma \in \{3,4\}\), inverted it to recover the cyclic string, and verified the successor-bispecial condition. For the rotation starting with \((\sigma - 1)^5\), direct computation of the terminated BWT gives \(r = r_c = \sigma^{k-1} + 2\) for these instances. Thus, using \(\chi = 2\sigma^{k-1} + 1\) (the argument of Lemma~\ref{lemma:smallest_set} applies verbatim at order \(k\)), the \(\sigma = 3\) instance gives \(r = r_c = 83\), \(\chi = 163\), and \(163/83 \approx 1.964\), while the \(\sigma = 4\) instance gives \(r = r_c = 258\), \(\chi = 513\), and \(513/258 \approx 1.988\).

The main open problem is to characterize the pairs \((\sigma,k)\) for which the lower bound \(\sigma^{k-1}+2\) is attainable. Equivalently, one may ask when the canonical boundary-merged candidate \(L_{\sigma,k}\) is LF-valid, or whether other cBWT patterns can attain the same bound. It is also open whether the \(k = 4\) obstruction is only a failure of \(L_{\sigma,k}\) or reflects a broader limitation on low-run constructions at order~\(4\).

%
%
\bibliographystyle{splncs04}
\bibliography{sources}

@book{lothaire1997,
  author = {M. Lothaire},
  title = {Combinatorics on Words},
  publisher = {Cambridge University Press},
  year = {1997}
}

@techreport{burrows1994,
  author = {Burrows, Michael and Wheeler, David},
  title = {A Block-sorting Lossless Data Compression Algorithm},
  institution = {Digital SRC},
  number = {124},
  year = {1994}
}

@book{lidl1997finite,
  title={Finite fields},
  author={Lidl, Rudolf and Niederreiter, Harald},
  number={20},
  year={1997},
  publisher={Cambridge University Press}
}

@article{higgins2012burrows,
  title={{B}urrows--{W}heeler transformations and de {B}ruijn words},
  author={Higgins, Peter M},
  journal={Theoretical Computer Science},
  volume={457},
  pages={128--136},
  year={2012},
  publisher={Elsevier}
}

@inproceedings{fici2025generalized,
  author    = {Fici, Gabriele and Gabory, Est{\'e}ban},
  title     = {Generalized {De} {B}ruijn Words, Invertible Necklaces, and the {B}urrows--{W}heeler Transform},
  booktitle = {50th International Symposium on Mathematical Foundations of Computer Science (MFCS 2025)},
  series    = {Leibniz International Proceedings in Informatics (LIPIcs)},
  volume    = {345},
  pages     = {48:1--48:18},
  publisher = {Schloss Dagstuhl -- Leibniz-Zentrum f{\"u}r Informatik},
  year      = {2025},
  doi       = {10.4230/LIPIcs.MFCS.2025.48}
}

@article{depuydt2023suffixient,
  title={Suffixient sets},
  author={Depuydt, Lore and Gagie, Travis and Langmead, Ben and Manzini, Giovanni and Prezza, Nicola},
  journal={arXiv preprint arXiv:2312.01359},
  year={2023}
}

@inproceedings{cenzato2024computing,
  author    = {Cenzato, Davide and Olivares, Francisco and Prezza, Nicola},
  title     = {On Computing the Smallest Suffixient Set},
  booktitle = {String Processing and Information Retrieval},
  series    = {Lecture Notes in Computer Science},
  volume    = {14899},
  pages     = {73--87},
  publisher = {Springer},
  year      = {2024},
  doi       = {10.1007/978-3-031-72200-4_6}
}

@inproceedings{navarro2025smallest,
  author    = {Navarro, Gonzalo and Romana, Giuseppe and Urbina, Cristian},
  title     = {Smallest Suffixient Sets as a Repetitiveness Measure},
  booktitle = {String Processing and Information Retrieval},
  series    = {Lecture Notes in Computer Science},
  volume    = {16073},
  pages     = {217--232},
  publisher = {Springer},
  year      = {2025},
  doi       = {10.1007/978-3-032-05228-5_18}
}

@inproceedings{mantaci2017burrows,
  title={{B}urrows--{W}heeler transform and run-length enconding},
  author={Mantaci, Sabrina and Restivo, Antonio and Rosone, Giovanna and Sciortino, Marinella},
  booktitle={International Conference on Combinatorics on Words},
  pages={228--239},
  year={2017},
  organization={Springer}
}

@misc{date2025neartightness,
  author        = {Date, Vinicius T. V. and Zatesko, Leandro M.},
  title         = {On the near-tightness of $\chi \leq 2r$: a general $\sigma$-ary construction and a binary case via {LFSR}s},
  year          = {2025},
  eprint        = {2512.20598},
  archivePrefix = {arXiv},
  note          = {To appear in LATIN 2026, Lecture Notes in Computer Science, Springer}
}

@inproceedings{likhomanov2011two,
  author    = {Likhomanov, Konstantin M. and Shur, Arseny M.},
  title     = {Two Combinatorial Criteria for {BWT} Images},
  booktitle = {Computer Science -- Theory and Applications},
  series    = {Lecture Notes in Computer Science},
  volume    = {6651},
  pages     = {385--396},
  publisher = {Springer},
  year      = {2011},
  doi       = {10.1007/978-3-642-20712-9_30}
}

\end{document}